\documentclass[12pt, a4paper]{amsart}

\usepackage{graphicx}        
\usepackage{tikz}
\usepackage{pgfplots}
\usepackage{amsmath}
\usepackage{amssymb}
\usepackage{mathptmx}
\usepackage{amsfonts, amsmath, amssymb,latexsym}
\usepackage[english]{babel}
\usepackage[utf8]{inputenc}
\usepackage{url}
\usepackage{mathptmx}

\usepackage{amsmath,amscd, amsfonts}

\newtheorem{thm}{Theorem}[section]

\newtheorem{prop}[thm]{Proposition}
\newtheorem{defn}[thm]{Definition}
\newenvironment{example}[1][Example]{\begin{trivlist}
\item[\hskip \labelsep {\bfseries #1}]}{\end{trivlist}}

\newcommand{\Pin}{\mathop{\mathrm{Pin}}}
\newcommand{\Spin}{\mathop{\mathrm{Spin}}}
\newcommand{\Cl}{\mathop{\mathrm{Cl}}}

\numberwithin{equation}{section}

\begin{document}


\title{The Birth of $E_8$ out of the Spinors of the Icosahedron}


\author{Pierre-Philippe Dechant}
\email[]{pierre-philippe.dechant@york.ac.uk}
\address{Departments of Mathematics and Biology, York Centre for Complex Systems Analysis, University of York, Heslington, York YO10 5GG, United Kingdom}


\date{\today}

\begin{abstract}
	$E_8$ is prominent in mathematics and theoretical physics, and is generally viewed as an exceptional symmetry  in an eight-dimensional space very different from the space we inhabit; for instance the Lie group $E_8$ features heavily in ten-dimensional superstring theory. 
	Contrary to that point of view, here we show that the $E_8$ root system can in fact be constructed from the icosahedron alone and can thus  be viewed purely in terms of three-dimensional geometry.
		The $240$ roots of $E_8$  arise in the 8D Clifford algebra of 3D space as a double cover of the $120$ elements of the icosahedral group, generated by the root system $H_3$. 
		As a by-product, by restricting to even products of root vectors (spinors) in the 4D even subalgebra of the Clifford algebra, 
		one can show that each 3D root system induces a root system in 4D, which turn out to also be exactly the exceptional 4D root systems. 
		The spinorial point of view explains their existence as well as their unusual automorphism groups. 
		This spinorial approach thus in fact allows one to construct all exceptional root systems within the geometry of three dimensions, which opens up a novel interpretation of these phenomena in terms of spinorial geometry.
\end{abstract}


\maketitle

\section{Introduction}
Lie groups are a central subject of 20th century mathematics as well as physics. In particular, the largest exceptional Lie group $E_8$ is central to String Theory and Grand Unified Theories and is thus arguably the single most important symmetry group in modern theoretical physics. Lie groups are continuous groups and are closely related to their corresponding Lie algebras. Their  non-trivial part in turn is described by a root system: a collection of reflection generating vectors called roots, which generate a reflection symmetry group (called a crystallographic Coxeter group or Weyl group). It is easy to move between those four related concepts and we will usually not make a distinction -- with the exception of non-crystallographic root systems such as  $H_3$ (which generates icosahedral symmetry) and its 4D analogue $H_4$, since their non-crystallographic nature means that there is no associated Lie algebra.  

This Lie-centric view point of much of theoretical physics, due to the fact that the gauge theories of high energy physics are formulated in terms of Lie groups, has unduly neglected the non-crystallographic groups. Here we instead advocate the root system as the more convenient and perhaps fundamental concept. The shift of perspective that allows to prove the central results of this paper is motivated in two ways. 
Firstly, non-crystallographic symmetries such as icosahedral symmetry are important in other areas of science outside the gauge theory paradigm, such as for the symmetries of viruses and fullerenes \cite{Crick:1956,Caspar:1962,DechantTwarockBoehm2011H3aff,DechantTwarockBoehm2011E8A4,Dechant:eo5029}, where we have  demonstrated the usefulness of a root system-based approach. 
The second ingredient is employing a Clifford algebra framework. 
Since for root systems the use of an inner product is already implicit, one can construct the Clifford algebra on the underlying vector space without any loss of generality, but achieving conceptual and computational simplifications. 
As these Clifford algebras have matrix representations, there is traditionally much prejudice towards resorting to matrices rather than working with the algebra directly; we advocate the more algebraic and geometric point of view, in particular as it provides a simple and geometric way to construct spinor groups. 

This approach offers a Clifford geometric construction of all the exceptional root systems as induced from the 3D root systems, a fact that has not been appreciated before from the Lie-theoretic perspective. 
The $E_8$ root system is  commonly thought of as an exceptional (i.e. there are no corresponding symmetry groups in arbitrary dimensions) phenomenon of  eight-dimensional geometry, far removed from 3D geometry. In this article we show that the eight dimensions of the 3D Clifford algebra actually allow $E_8$ to be unveiled as a 3D geometric phenomenon in disguise; likewise all 4D exceptional root systems also arise within 3D geometry in this way. This opens up a completely novel way of viewing exceptional higher-dimensional phenomena in terms of 3D spinorial geometry.

This paper is structured in the following way. After some preliminary definitions and background in Sections \ref{background} (root systems and reflection groups) and \ref{sec_versor} (Clifford algebras), we show in general that any 3D symmetry group induces a 4D symmetry group via their root systems (Section \ref{sec_pin}). In particular,  the Platonic root systems $(A_3, B_3, H_3)$  induce all the exceptional 4D root systems $(D_4, F_4, H_4)$ in terms of 3D spinors; this also explains their unusual symmetry groups.  In Section \ref{H4} we concretely explain the case of icosahedral symmetry $H_3$ inducing the exceptional largest non-crystallographic  Coxeter group $H_4$ from a spinor group (the binary icosahedral group $2I$) that describes the $60$ icosahedral rotations in terms of $120$ spinors doubly covering the rotations. The above collection of Platonic root systems $(A_3, B_3, H_3)$ in fact forms a trinity that is related to the trinity of exceptional Lie groups $(E_6, E_7, E_8)$ via various intermediate trinities and also via my new spinor construction in combination with the McKay correspondence. This is the first hint that the icosahedron may be indirectly related with $E_8$. Section \ref{birth} makes a completely new, direct connection between them by concretely constructing the $240$ roots of the $E_8$ root system from the $240$ pinors that doubly cover the $120$ icosahedral reflections and rotations in the eight-dimensional Clifford algebra of 3D. Thus \textbf{all} the exceptional root systems can  in fact be seen as induced from the polyhedral symmetries and the Clifford algebra of 3D. This offers a new way of better understanding these exceptional phenomena in terms of spinorial geometry; this reinterpretation has the potential for a wide range of profound consequences. We conclude in Section \ref{concl}.

\section{Root systems and reflection groups}\label{background}

In this section, we introduce Coxeter (reflection) groups via the  root systems that generate them \cite{Humphreys1990Coxeter}:

\begin{defn}[Root system] \label{DefRootSys}
A \emph{root system} is a collection $\Phi$ of non-zero (root)  vectors $\alpha$ spanning an $n$-dimensional Euclidean vector space $V$ endowed with a positive definite bilinear form denoted by $(\cdot \vert \cdot)$, which satisfies the  two axioms:
\begin{enumerate}
\item $\Phi$ only contains a root $\alpha$ and its negative, but no other scalar multiples: $\Phi \cap \mathbb{R}\alpha=\{-\alpha, \alpha\}\,\,\,\,\forall\,\, \alpha \in \Phi$. 
\item $\Phi$ is invariant under all reflections $$s_\alpha: \lambda\rightarrow s_\alpha(\lambda)=\lambda - 2\frac{(\lambda|\alpha)}{(\alpha|\alpha)}\alpha$$ corresponding to reflections in hyperplanes orthogonal to root vectors in $\Phi$: $s_\alpha\Phi=\Phi \,\,\,\forall\,\, \alpha\in\Phi$. 
\end{enumerate}
\end{defn}

For a crystallographic root system, a subset $\Delta$ of $\Phi$, called \emph{simple roots} $\alpha_1, \dots, \alpha_n$, is sufficient to express every element of $\Phi$ via $\mathbb{Z}$-linear combinations with coefficients of the same sign. 
$\Phi$ is therefore  completely characterised by this basis of simple roots. In the case of the non-crystallographic root systems $H_2$, $H_3$ and $H_4$, the same holds for the extended integer ring $\mathbb{Z}[\tau]=\lbrace a+\tau b| a,b \in \mathbb{Z}\rbrace$, where $\tau$ is   the golden ratio $\tau=\frac{1}{2}(1+\sqrt{5})=2\cos{\frac{\pi}{5}}$, and $\sigma$ is its Galois conjugate $\sigma=\frac{1}{2}(1-\sqrt{5})$ (the two solutions to the quadratic equation $x^2=x+1$).  
 For the crystallographic root systems, the classification in terms of Dynkin diagrams essentially follows the one familiar from Lie groups and Lie algebras, as their Weyl groups are  the crystallographic Coxeter groups. A mild generalisation to so-called Coxeter-Dynkin diagrams is necessary for the non-crystallographic root systems:
\begin{defn}[Coxeter-Dynkin diagram and Cartan matrix] 
	A graphical representation of the geometric content of a root system is given by \emph{Coxeter-Dynkin diagrams}, where nodes correspond to simple roots, orthogonal roots are not connected, roots at $\frac{\pi}{3}$ have a simple link, and other angles $\frac{\pi}{m}$ have a link with a label $m$. 
The  \emph{Cartan matrix} of a set of simple roots $\alpha_i\in\Delta$ is defined as the matrix
	$A_{ij}=2{(\alpha_i\vert \alpha_j)}/{(\alpha_j\vert \alpha_j)}$.
\end{defn}
For instance, the root system of the icosahedral group $H_3$ has one link labelled by $5$ (via the above relation $\tau=2\cos{\frac{\pi}{5}}$), as does its four-dimensional analogue $H_4$.

The reflections in the second axiom of the root system generate a reflection group. A Coxeter group is a mathematical abstraction of the concept of a reflection group via  involutory  generators (i.e. they square to the identity, which captures the idea of a reflection), subject to mixed relations that represent $m$-fold rotations (since two successive reflections generate a rotation in the plane defined by the two roots).
\begin{defn}[Coxeter group] A \emph{Coxeter group} is a group generated by a set of involutory generators $s_i, s_j \in S$ subject to relations of the form $(s_is_j)^{m_{ij}}=1$ with $m_{ij}=m_{ji}\ge 2$ for $i\ne j$. 
\end{defn}
The  finite Coxeter groups have a geometric representation where the involutions are realised as reflections at hyperplanes through the origin in a Euclidean vector space $V$, i.e. they are essentially  just the classical reflection groups. In particular, then the abstract generator $s_i$ corresponds to the simple {reflection}
$s_i: \lambda\rightarrow s_i(\lambda)=\lambda - 2\frac{(\lambda|\alpha_i)}{(\alpha_i|\alpha_i)}\alpha_i$
 in the hyperplane perpendicular to the  simple {root } $\alpha_i$.
The action of the Coxeter group is  to permute these root vectors, and its  structure is thus encoded in the collection  $\Phi\in V$ of all such roots, which in turn form a root system.

Here we employ a Clifford algebra framework, which via the geometric product affords a uniquely simple prescription for performing reflections via `sandwiching' $-\alpha \lambda \alpha$ (assuming unit normalisation).  Since due to the Cartan-Dieudonn\'e theorem any orthogonal transformation in any dimension and signature can be written as products of reflections, the `sandwiching' prescription  in fact extends to any orthogonal transformation. For any root system, the quadratic form mentioned in the definition can always be used to enlarge the $n$-dimensional vector space $V$ to the corresponding $2^n$-dimensional Clifford algebra. The Clifford algebra  is therefore a very natural object to consider in this context, as its unified structure simplifies many problems both conceptually and computationally, rather than applying the linear structure of the space and the inner product separately. In particular, it provides (s)pinor double covers of the (special) orthogonal transformations, as well as geometric quantities that serve as unit imaginaries. We therefore advocate the use of Clifford algebra as the most natural framework for root systems and reflection groups.

\section{Clifford versor framework}\label{sec_versor}

Clifford algebra can be viewed as a deformation of the (perhaps more familiar) exterior algebra by a quadratic form -- though we do not necessarily advocate this point of view; they are isomorphic as vector spaces, but not as algebras, and Clifford algebra is in fact much richer due to the invertibility of the algebra product. 
The \emph{geometric product} of Geometric/Clifford Algebra is defined by $xy=x\cdot y+x \wedge y$, where the scalar product (given by the symmetric bilinear form) is the symmetric part $x\cdot y=(x|y)=\frac{1}{2}(xy+yx)$ and the exterior product the antisymmetric part $x\wedge y=\frac{1}{2}(xy-yx)$   \cite{Hestenes1966STA, HestenesSobczyk1984, Hestenes1990NewFound,LasenbyDoran2003GeometricAlgebra}.
It provides a very compact and efficient way of handling reflections in any number of dimensions, and thus by the \emph{Cartan-Dieudonn\'e theorem} in fact of any orthogonal transformation. For a unit vector $\alpha$, the two terms in the above formula for a reflection of a vector $v$ in the hyperplane orthogonal to $\alpha$ simplify to the double-sided (`sandwiching') action of $\alpha$ via the geometric product
	\begin{equation}\label{in2refl}
	  v\rightarrow s_\alpha v=v'=v-2(v|\alpha)\alpha=v-2\frac{1}{2}(v\alpha+\alpha v)\alpha=v-v\alpha^2-\alpha v\alpha=-\alpha v \alpha.
	\end{equation}
This  prescription for reflecting vectors in hyperplanes is remarkably compact (note that $\alpha$ and $-\alpha$ encode the same reflection and thus provide a double cover). Via the Cartan-Dieudonn\'e theorem, any orthogonal transformation can be written as the product of reflections,  and thus by performing consecutive reflections, each given via `sandwiching', one is led to define a versor as
 a Clifford multivector $A=a_1a_2\dots a_k$, that is the product of $k$ unit vectors $a_i$  \cite{Hestenes1990NewFound}.  Versors form a multiplicative group called the versor/pinor group $\Pin$ under single-sided multiplication with the geometric product, with inverses given by $\tilde{A}A=A\tilde{A}=\pm 1$, where the tilde denotes the reversal of the order of the constituent vectors $\tilde{A}=a_k\dots a_2a_1$, and  the $\pm$-sign defines its parity.
Every orthogonal transformation $\underbar{A}$ of a vector $v$ can thus be expressed by means of unit versors/pinors via
\begin{equation}\label{in2versor}
\underbar{A}: v\rightarrow  v'=\underbar{A}(v)=\pm{\tilde{A}vA}.
\end{equation}
Unit versors are double-covers of the respective orthogonal transformation, as $A$ and $-A$ encode the same transformation. Even versors $R$, that is, products of an even number of vectors, are called \emph{spinors} or \emph{rotors}. They form a subgroup of the $\Pin$ group and constitute a double cover of the special orthogonal group, called the $\Spin$ group.
Clifford algebra therefore provides a particularly natural and simple construction of the $\Spin$ groups. Thus the remarkably simple construction of the binary polyhedral groups in Section \ref{sec_pin} is not at all surprising from a Clifford point of view, but appears to be unknown in the Coxeter community, and ultimately leads to the novel result of the spinor induction theorem of (exceptional) 4D root systems in Section \ref{sec_pin} and the construction of the $E_8$ root system from $H_3$ in Section \ref{birth}.

\begin{table}
\caption{Versor framework for a unified treatment of the chiral, full,  binary and pinor polyhedral groups}
\label{tab:1}       
%
%
\begin{tabular}{p{2.6cm}p{3.3cm}p{4.9cm}}
\hline
Continuous group &Discrete subgroup & Multivector action  \\
\hline
$SO(n)$&rotational/chiral & $x\rightarrow \tilde{R}xR$\\
$O(n)$&reflection/full & $x\rightarrow \pm\tilde{A}xA$\\
$\Spin(n)$&binary  & spinors $R$ under $(R_1,R_2)\rightarrow R_1R_2$\\
$\Pin(n)$& pinor & pinors $A$ under $(A_1,A_2)\rightarrow A_1A_2$\\
\hline
\end{tabular}
\end{table}

The versor realisation of the orthogonal group is much simpler than conventional matrix approaches. 
Table \ref{tab:1} summarises the various action mechanisms of multivectors: a rotation (e.g. the continuous group $SO(3)$ or the discrete subgroup, the chiral icosahedral group $I=A_5$) is given by double-sided action of a spinor $R$, whilst these spinors themselves form a group under single-sided action/multiplication (e.g. the continuous group $\Spin(3)\sim SU(2)$ or the discrete subgroup, the binary icosahedral group $2I$).
Likewise, a reflection (continuous $O(3)$ or the discrete subgroup, the full icosahedral group the Coxeter group $H_3$) corresponds to sandwiching with the versor $A$, whilst the versors single-sidedly form a multiplicative group (the $\Pin(3)$ group or the discrete analogue, the double cover of $H_3$, which we denote $\Pin(H_3)$). In the conformal geometric algebra setup one uses the fact that the conformal group   $C(p,q)$ is homomorphic to $SO(p+1,q+1)$ to treat translations as well as rotations in a unified versor framework  \cite{HestenesSobczyk1984, LasenbyDoran2003GeometricAlgebra,Dechant2011Thesis, Dechant2012AGACSE, Dechant2015ICCA}. \cite{Dechant2012AGACSE, Dechant2015ICCA} also discuss reflections, inversions, translations and modular transformations in this way. 

\begin{example}
The Clifford/Geometric algebra of three dimensions $\Cl(3)$ is generated by three orthogonal -- and thus anticommuting -- unit vectors $e_1$, $e_2$ and $e_3$. It also contains the three bivectors $e_1e_2$, $e_2e_3$ and $e_3e_1$ that all square to $-1$, as well as the  highest grade object $e_1e_2e_3$   (trivector and pseudoscalar), which also squares to $-1$. Therefore, in Clifford algebra various geometric objects arise that provide imaginary structures; however, there can be different ones and they can have non-trivial commutation relations with the rest of the algebra. 
\begin{equation}\label{in2PA}
  \underbrace{\{1\}}_{\text{1 scalar}} \,\,\ \,\,\,\underbrace{\{e_1, e_2, e_3\}}_{\text{3 vectors}} \,\,\, \,\,\, \underbrace{\{e_1e_2=Ie_3, e_2e_3=Ie_1, e_3e_1=Ie_2\}}_{\text{3 bivectors}} \,\,\, \,\,\, \underbrace{\{I\equiv e_1e_2e_3\}}_{\text{1 trivector}}.
\end{equation}
\end{example}

\section{The general 4D spinor induction construction}\label{sec_pin}

In this section we prove that any 3D root system yields a 4D root system via the spinor group obtained by multiplying root vectors in the Clifford algebra  \cite{Dechant2012Induction}.
\begin{prop}[$O(4)$-structure of spinors]\label{HGA_O4}
The space of $\Cl(3)$-spinors $R=a_0+a_1e_2e_3+a_2e_3e_1+a_3e_1e_2$ can be endowed with a \emph{4D Euclidean norm} $|R|^2=R\tilde{R}=a_0^2+a_1^2+a_2^2+a_3^2$ induced by the  \emph{inner product} $(R_1,R_2)=\frac{1}{2}(R_1\tilde{R}_2+R_2\tilde{R}_1)$ between  two spinors $R_1$ and $R_2$. 
\end{prop}
This allows one to  reinterpret the group of 3D spinors generated from a 3D root system as a set of 4D vectors, which in fact can be shown to  satisfy the axioms of a root system as given in Definition \ref{DefRootSys}. 
\begin{thm}[Induction Theorem]\label{HGA_4Drootsys}
Any 3D root system gives rise to a spinor group $G$ which induces a root system in 4D.
\end{thm}
\begin{proof}
Check the two axioms for the root system $\Phi$ consisting of the set of 4D vectors given by the 3D spinor group:
\begin{enumerate}
\item By construction, $\Phi$ contains the negative of a root $R$ since spinors provide a double cover of rotations, i.e. if $R$ is in a spinor group $G$, then so is $-R$ , but no other scalar multiples (normalisation to unity). 
\item $\Phi$ is invariant under all reflections with respect to the inner product $(R_1,R_2)$ in Proposition \ref{HGA_O4} since $R_2'=R_2-2(R_1, R_2)/(R_1, {R}_1) R_1=-R_1\tilde{R}_2R_1\in G$ for $R_1, R_2 \in G$ by the closure property of the group $G$ (in particular $-R$ and $\tilde{R}$ are in $G$ if $R$ is). 
\end{enumerate}
\end{proof}

Since the number of irreducible 3D root systems is limited to $(A_3, B_3, H_3)$, this yields a definite list of induced root systems in 4D -- this turns out to be  $(D_4, F_4, H_4)$, which are exactly the exceptional root systems in 4D. In fact both sets of three are trinities: named after Arnold's observation that many related objects in mathematics form sets of three, beginning with the trinity $(\mathbb{R},\mathbb{C},\mathbb{H})$ \cite{Arnold1999symplectization,Arnold2000AMS}, and extending to projective spaces, Lie algebras, spheres, Hopf fibrations etc. Arnold's original link between these two trinities $(A_3, B_3, H_3)$ and $(D_4, F_4, H_4)$ that arise here is extremely convoluted, and our construction presents a novel  direct link between the two. 
 The abundance of root systems in 4D can in some sense be thought of as due to the accidentalness of our spinor construction.  In particular, the induced root systems are precisely the exceptional (i.e. they do not have counterparts in other dimensions) root systems in 4D: $D_4$ has the  triality symmetry (permutation symmetry of the three legs in the diagram) that is exceptional in 4D and is of great importance in string theory, showing the equivalence of the Ramond-Neveu-Schwarz and the Green-Schwarz strings. $F_4$ is the only $F$-type root system, and $H_4$ is the largest non-crystallographic root system. In contrast, in arbitrary dimensions there are only   $A_n$ ($n$-simplex),  $B_n$  ($n$-hypercube and $n$-hyperoctahedron) and $D_n$.

Not only is there an abundance of root systems related to the Platonic solids as well as their exceptional nature  \cite{Dechant2013Platonic}, but they also have unusual automorphism groups, in that the order of the groups is proportional to the square of the number of roots. This is also explained via the above spinor construction via the following result: 
\begin{thm}[Spinorial symmetries]\label{HGAsymmetry}
A root system induced through the Clifford spinor construction via a binary polyhedral spinor group $G$ has an automorphism group that trivially contains two factors of the respective spinor group $G$ acting from the left and from the right.
\end{thm}
This systematises many case-by-case observations on the structure of the automorphism groups \cite{Koca2006F4,Koca2003A4B4F4}, and shows that all of the 4D geometry is already contained in 3D \cite{Dechant2012CoxGA}. For instance, the automorphism group of the $H_4$ root system is $2I\times 2I$ -- in the spinor picture, it is not surprising that $2I$ yields both the root system and the two factors in the automorphism group. We therefore consider the example of the induction $H_3\rightarrow H_4$ in more detail, as this will prove one of the crucial steps towards constructing the $E_8$ root system. 

\section{Example: Constructing $H_4$ from $H_3$}\label{H4}

 Here we  construct the spinor group generated by the simple reflections of $H_3$ explicitly as an example and as an intermediate result. The simple roots of $H_3$ are taken as $$\alpha_1=e_2, \alpha_2=-\frac{1}{2}((\tau -1)e_1+e_2+\tau e_3)\text{ and }\alpha_3=e_3.$$ Under free multiplication under the Clifford algebra product, they generate a group with $240$  pinors, whilst the even subgroup consists of $120$ spinors, for instance of the form $$\alpha_1\alpha_2=-\frac{1}{2}(1-(\tau -1)e_1e_2+\tau e_2e_3)\text{ and }\alpha_2\alpha_3=-\frac{1}{2}(\tau-(\tau -1)e_3e_1+e_2e_3).$$ These are the double covers of  $H_3=A_5\times \mathbb{Z}_2$ of order $120$ and $I=A_5$ of order $60$, respectively. The spinors have four components ($1$, $e_1e_2$, $e_2e_3$, $e_3e_1$); by taking the components of these $120$ spinors as a set of vectors in 4D, one obtains $$(\pm 1, 0, 0, 0) \text{ (8 permutations) }$$
$$\frac{1}{2}(\pm 1, \pm 1, \pm 1, \pm 1) \text{ (16 permutations) }$$
$$\frac{1}{2}(0, \pm 1, \pm \sigma, \pm \tau) \text{ (96 even permutations) },$$
which are precisely the $120$ roots of the  $H_4$ root system. This is very surprising from a Coxeter perspective, as one usually thinks of $H_3$ as a subgroup of $H_4$, and therefore of $H_4$ as more `fundamental'; however,  one now sees that $H_4$ does not in fact contain any structure that is not already contained in $H_3$, and can therefore think of $H_3$ as more fundamental \cite{Dechant2012CoxGA}. 

From a Clifford perspective it is not surprising to find this group of $120$ spinors, which is the binary icosahedral group $2I$, since Clifford algebra provides a simple construction of the Spin groups; however, this is completely unexpected from the conventional Coxeter and Lie group perspective. This spinor group $2I$ has $120$ elements and $9$ conjugacy classes.
 $I$  has five conjugacy classes  and it being of order $60$ implies that it has five irreducible representations of dimensions  $1$, $3$, $\bar{3}$, $4$ and $5$ (since the sum of the dimensions squared gives the order of the group $\sum d_i^2=|G|$). The nine conjugacy classes of the binary icosahedral group $2I$ of order $120$ mean that this acquires a further four irreducible spinorial representations $2_s$, $2_s'$, ${4_s}$ and ${6_s}$. The automorphism group of the $H_4$ root system is  $2I\times 2I$, which in this framework is trivial (see Theorem \ref{HGAsymmetry}), as it merely reflects that $2I$ is closed with respect to left and right multiplication, but the literature has been very confused about this simple fact by overlooking the underlying simple construction.  It is also worth pointing out that it is convenient to have all these 4 different types of polyhedral groups in a unified framework within the Clifford algebra, rather than using $SO(3)$ matrices for the rotations and then having to somehow move to $SU(2)$ matrices for the binary groups. 

In terms of a quaternionic description of the even (spinorial) subalgebra (they are isomorphic, though we advocate here the geometric rather than algebraic point of view), the $H_3$ root system consists precisely of the pure quaternions, i.e. those without a real part, and the full group can be generated from those under quaternion multiplication. This is very poorly understood in the literature, and just hinges on the above description in terms of spinors together with the fact that the inversion  $\pm I$ is contained in the full $H_3$ group, as one can then trivially Hodge dualise root vectors to bivectors, which are pure quaternions. We have explained this in previous work \cite{Dechant2012CoxGA}. For instance, the statement is not true when the inversion is not contained in the group, as is the case for $A_3$. However, the spinorial induction construction still works for this, yielding $D_4$. Moreover, the `quaternionic generators' generating the 4D groups via quaternion multiplication are just seen to be the even products of 3D simple roots $\alpha_1\alpha_2$ and $\alpha_2\alpha_3$ so that the 4D group manifestly does not contain anything that was not already contained in the 3D group. We therefore believe the spinor induction point of view of going from 3D to 4D is more fundamental than the pure quaternion approach identifying the 3D group as a subgroup of the 4D group.

This binary icosahedral group has a curious connection with the affine Lie algebra $E_8^+$ (and likewise for the other binary polyhedral groups and the affine Lie algebras of $ADE$-type) via the so-called McKay correspondence \cite{Mckay1980graphs}, which is twofold: 
 We can define a graph by assigning a node to each of the nine irreducible representations of the binary icosahedral group with the following rule for connecting nodes: each node corresponding to a certain irreducible representation is connected to the nodes corresponding to those irreducible representations that are contained in its tensor product with the irrep $2_s$. For instance, tensoring the trivial representation $1$ with $2_s$ trivially gives $2_s$ and thus the only link  $1$ has is with $2_s$; $2_s\otimes 2_s=1+3$, such that $2_s$ is connected to $1$ and $3$, etc. The graph that is built up in this way is precisely the Dynkin diagram of affine $E_8$, as shown in Figure \ref{figE6aff}. The second connection is the following observation: the Coxeter element is the product of all the simple reflections $\alpha_1\dots \alpha_8$ and its order, the Coxeter number $h$, is $30$ for $E_8$. This also happens to be the sum of the dimensions of the irreducible representations of $2I$, $\sum d_i$. This extends to all other cases of binary polyhedral groups and $ADE$-type affine Lie algebras.

\begin{figure}
	\begin{center}
\includegraphics[width=8cm]{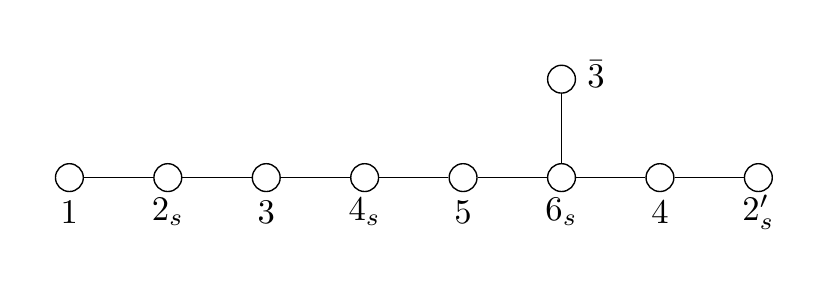}
\end{center}
\caption[$E_6^+$]{The graph depicting the tensor product structure of the binary icosahedral group $2I$ is the same as the Dynkin diagram of  affine  $E_8$. }
\label{figE6aff}
\end{figure}

The following indirect connection between $(A_3, B_3, H_3)$ and $(E_6, E_7, E_8)$ via Clifford spinors does not seem to be known: we note that $(12, 18, 30)$ are not only the sums of the dimensions of the irreducible representations $\sum d_i$ of the binary polyhedral groups, but more fundamentally are exactly the numbers of roots $\Phi$ in the 3D root systems $(A_3, B_3, H_3)$ from which these binary polyhedral groups are generated. This fact connects the 3D root systems through the binary polyhedral groups and via the McKay correspondence all the way to the affine Lie algebras. Our construction therefore makes deep connections between trinities, and puts the McKay correspondence into a wider framework. It is also striking that the affine Lie algebra and the 4D root system trinities have identical Dynkin diagram symmetries: $D_4$ and $E_6^+$ have triality $S_3$, $F_4$ and $E_7^+$ have an $S_2$ symmetry and $H_4$ and $E_8^+$ only have $S_1$.  There is thus an indirect connection between the icosahedral group $H_3$ and the exceptional $E_8$. In the next section we will show a new, explicit direct connection between $H_3$ and $E_8$ within the Clifford algebra of 3D, by identifying the $240$ roots of $E_8$ with the $240$ pinors of icosahedral symmetry.

\section{The birth of $E_8$ out of the spinors of the icosahedron}\label{birth}

Previously, we discussed the construction of the $120$ elements of the binary icosahedral group $2I$ (within the 4D even subalgebra of the 8D Clifford algebra of 3D space), which can be reinterpreted as the $120$ roots of $H_4$. Here we list these elements again in terms of a 4D basis
$$(\pm 1, 0, 0, 0) \text{ (8 permutations) }$$
$$\frac{1}{2}(\pm 1, \pm 1, \pm 1, \pm 1) \text{ (16 permutations) }$$
$$\frac{1}{2}(0, \pm 1, \pm \sigma, \pm \tau) \text{ (96 even permutations) }.$$
A convenient set of simple roots for $H_4$ is given by $$\alpha_1=\frac{1}{2}(-\sigma, -\tau, 0, -1), \alpha_2=\frac{1}{2}(0, -\sigma, -\tau,  1),$$ $$\alpha_3=\frac{1}{2}(0,1, -\sigma, -\tau)\text{ and }\alpha_4=\frac{1}{2}(0, -1, -\sigma, \tau).$$

Since the $H_3$ root system contains three orthogonal roots, e.g.  $(1, 0, 0)$, $(0, 1, 0)$ and $(0, 0, 1)$, the pinor group generated under free Clifford multiplication contains the inversion, given by double-sided action of $\pm e_1e_2e_3=\pm I$. So whilst the $120$ even products of root vectors stay in the 4D even spinor  subalgebra consisting of scalar and bivector parts,  the $120$ odd products in the remaining vector and pseudoscalar part of the algebra are  just a second copy of these spinors multiplied by $I$. So for the pinors double covering the group $H_3$ of order $120$ in the  Clifford algebra of 3D space -- which is itself an 8D vector space -- one gets $240$ objects, as one would expect for a construction of $E_8$, and furthermore, they consist of a copy of $H_4$ with $120$ roots and another copy multiplied by $I$, schematically $H_4+IH_4$.

For this set of $240$ pinors one can now define a `reduced inner product' in the following way  \cite{Wilson1986E8, MoodyPatera:1993b}: we keep the spinor copy of $H_4$ and multiply the copy $IH_4$ by $\tau I$, then take inner products taking into account the recursion relation $\tau^2=\tau+1$ but finally in this inner product (counterintuitively) setting $\tau$ equal to zero.  
An equivalent point of view is this:  a general inner product between roots over the extended integer ring $\mathbb{Z}[\tau]$ is a $\mathbb{Z}[\tau]$-integer and thus has an integer and a $\tau$ part $(\cdot, \cdot) = a+\tau b$. The reduced inner product $(\cdot, \cdot)_\tau$ is thus now defined as the integer part alone $$(\cdot, \cdot)_\tau=(a+\tau b)_\tau:=a.$$ 

\begin{figure}
	\begin{center}
 \includegraphics[width=8cm]{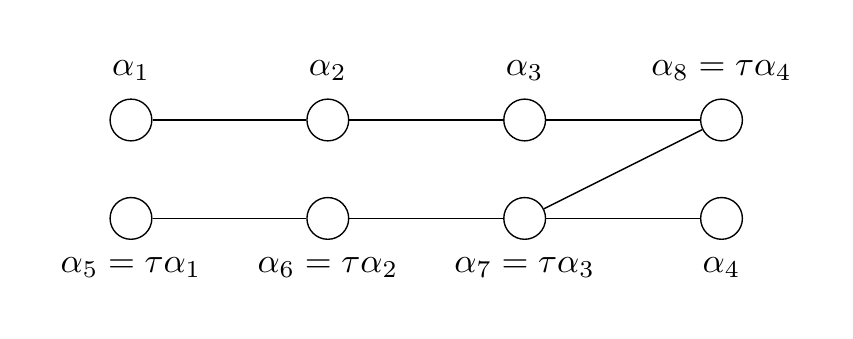}    
  \caption[$E_8$]{$E_8$ Coxeter-Dynkin diagram via two copies of $H_4$ and the reduced scalar product.}
\label{figE8}
\end{center}
\end{figure}

This set of $240$ Clifford pinors contains the above mentioned choice of simple roots of $H_4$ 
$$\alpha_1=\frac{1}{2}(-\sigma, -\tau, 0, -1), \alpha_2=\frac{1}{2}(0, -\sigma, -\tau,  1),$$ $$\alpha_3=\frac{1}{2}(0,1, -\sigma, -\tau) \text { and } \alpha_4=\frac{1}{2}(0, -1, -\sigma, \tau)$$ 
along with their $\tau$-multiples 
$$\alpha_5=\tau \alpha_1=\frac{1}{2}(1, -\tau-1, 0, -\tau), \alpha_6=\tau \alpha_2=\frac{1}{2}(0, 1, -\tau-1,  \tau),$$
$$\alpha_7=\tau \alpha_3=\frac{1}{2}(0,\tau, 1, -\tau-1) \text{ and } \alpha_8=\tau \alpha_4=\frac{1}{2}(0, -\tau, 1, \tau+1).$$ 
With the reduced inner product, most of the inner products among the two sets of $H_4$ are not affected, but crucially, there are a few products that change links in and between the two $H_4$ diagrams, turning them into the $E_8$ diagram. In turn, the $240$ icosahedral pinors  are the $240$ roots of $E_8$ with respect to the reduced inner product, and it is straightforward (if tedious) to check closure under reflections with respect to this new inner product. 

We discuss a few exemplary and crucial inner products to illustrate the method of taking the reduced inner product as well as the result. 
$$(\alpha_1,\alpha_2)_\tau=\left(-\frac{1}{2}\right)_\tau=-\frac{1}{2},$$
which illustrates that the simply-connected nodes of the original $H_4$ diagrams are not affected. 
$$(\alpha_3,\alpha_4)=\frac{1}{4}(-1+\sigma^2-\tau^2)=\frac{1}{4}(-1+\sigma+1-\tau-1)=\frac{1}{4}(-1+1-\tau-\tau)=-\frac{\tau}{2},$$ resulting in the link labelled by $5$ in the $H_4$ diagram. However, 
$$(\alpha_3,\alpha_4)_\tau=\left(-\frac{\tau}{2}\right)_\tau=0,$$
such that  these roots are now orthogonal with respect to the new reduced inner product. On the other hand,
$$(\alpha_7,\alpha_8)_\tau=(\tau \alpha_3,\tau \alpha_4)_\tau=\left(-\frac{\tau^3}{2}\right)_\tau=\left(-\frac{2\tau+1}{2}\right)_\tau=-\frac{1}{2},$$
such that the $5$-labelled link from the other $H_4$ diagram is now actually just a simple link. 
The surprising additional (simple) links introduced are instead
$$(\alpha_7,\alpha_4)_\tau=(\tau \alpha_3,\alpha_4)_\tau=(\alpha_3,\tau \alpha_4)_\tau=(\alpha_8,\alpha_3)_\tau=-\left(\frac{\tau^2}{2}\right)_\tau=-\left(\frac{\tau+1}{2}\right)_\tau=-\frac{1}{2}.$$
 The Cartan matrix for this set of simple roots is thus
$$ \begin{pmatrix}
   2&-1&0&0&0&0&0&0  
\\ -1&2&-1&0&0&0&0  &0
\\ 0&-1&2&0&0&0&0&-1
\\ 0&0&0&2&0&0&-1&0  
\\ 0&0&0&0&2&-1 &0&0 
\\ 0&0&0&0&-1&2 &-1 &0
\\ 0&0&0&-1&0&-1&2 &-1
\\ 0&0&-1&0&0&0&-1 &2\end{pmatrix},$$
which is the $E_8$ Cartan matrix in slightly unusual form, and likewise for the corresponding Coxeter diagram  as shown in Fig. \ref{figE8}. 
This completes the construction of $E_8$ from the icosahedral root system $H_3$ within the eight-dimensional Clifford algebra of 3D space. 

Surprisingly, the $E_8$ root system has therefore been lurking in plain sight within the geometry of the Platonic solid the icosahedron for three millennia. As with the 4D induction construction, this discovery was only possible in Clifford algebra -- there is much prejudice against the usefulness of Clifford algebras (since they have matrix representations) and usually matrix methods are equivalent if less insightful -- but the 4D and 8D induction constructions are to our knowledge the only results that \emph{require} Clifford algebra and were  invisible to standard matrix methods.

\section{Conclusion}\label{concl}
We have constructed for the first time the $E_8$ root system  from the icosahedral root system $H_3$ within the 3D Clifford algebra, which is itself an eight-dimensional vector space. 
Some of the details were previously known individually, but the  construction of $E_8$ from the icosahedron is  new and puts those details into a coherent framework.
It is fascinating that the exotic $E_8$ root system that is so central to modern theoretical physics and mathematics has been hiding in plain sight since at least the ancient Greek times of Plato, lurking in the shadows of the icosahedron, and can be viewed much more naturally as a phenomenon within 3D geometry. This is completely contrary to the prevailing view of $E_8$ as an eight-dimensional phenomenon with no connection to the space we inhabit. The same holds for all the 4D exceptional root systems -- their construction, existence and automorphism groups  are thought of much more naturally in terms of spinorial geometry. 
These spinorial constructions therefore in fact yield  all the exceptional root systems and could have  profound consequences by opening up a  new area of interpreting these phenomena in terms of spinorial geometry.

\section*{Acknowledgements}
{I would like to thank Reidun Twarock, Anne Taormina, C\'eline Boehm, David Hestenes, Anthony Lasenby, John Stillwell and Robert Wilson for helpful discussions, criticism and support over the years. }


\end{document}